\documentclass[11pt]{article}
\usepackage[ansinew]{inputenc}

\usepackage{latexsym,bezier,enumerate,longtable,dcolumn,color,pstcol,xspace,curves,mathpazo,url}
\usepackage{amsmath,amsthm,amsopn,amstext,amscd,amsfonts,amssymb,fancybox,pst-node,pst-tree}

\usepackage{epsfig,epic,graphicx,pstricks}

\newcommand{\V}[1]{\mathbf{#1}}

\theoremstyle{plain}
\newtheorem{proposition}{\bf Proposition}
\newtheorem{theorem}{\bf Theorem}

\newtheorem{corollary}{\bf Corollary}
\theoremstyle{definition}
\newtheorem{definition}{\bf Definition}
\newtheorem{example}{\bf Example}

\begin{document}

\title{Some game theoretic marketing attribution models
\thanks{This research has been supported by I+D+i research project MTM2015-70550-P from the Government of Spain.}
}


\maketitle

\vspace*{-1cm} {\baselineskip .2in

\begin{center}
{\small
\begin{tabular}{l}
\bf Elisenda Molina
\\ Departamento de Estad\'{\i}stica, Universidad Carlos III de Madrid, Spain
\\ e-mail: \url{elisenda.molina@uc3m.es}
\\
\bf Juan Tejada
\\ Instituto de Matem\'atica Interdisciplinar (IMI), Departamento de Estad\'{\i}stica e Investigaci\'on Operativa,
\\ Universidad Complutense de Madrid, Spain
\\  e-mail: \url{jtejada@mat.ucm.es}
\\
 \bf Tom Weiss
\\ Deductive Inc., US
\\ e-mail: \url{tom@deductive.com}
\end{tabular}}
\end{center}
}

\begin{abstract}
In this paper, we propose and analyse two game theoretical models useful to design marketing channels attribution mechanisms based on cooperative TU games and bankruptcy problems, respectively. First, we analyse the {\it Sum Game}, a coalitional game introduced by Morales (2016). We extend the ideas introduced in Zhao et al. (2018) and Cano-Berlanga et al. (2017) to the case in which the order and the repetition of channels on the paths to conversion are taken into account. In all studied cases, the Shapley value is proposed as the attribution mechanism. Second, a bankruptcy problem approach is proposed, and a similar analysis is developed relying on the Constrained Equal Loss (CEL) and Proportional (PROP) rules as attribution mechanisms. In particular, it is relevant to note that the class of attribution bankruptcy problems is a proper subclass of bankruptcy problems.

\end{abstract}

\section{Introduction}
\label{intro}

The attribution of the obtained benefits to the different channels involved in a marketing campaign is a relevant issue because it can help to optimally assign the marketing budget and, in general, to have a deep knowledge of the effects of a campaign. Therefore there is a vast literature on this problem in the field of marketing that we do not try to analyse here. The reader is referred to Jayawardane et al. (2015) and Choi (2020) for a review and classification of the methodologies considered in this field.

Taking into account that an attribution problem is in essence a benefit allocation problem, we are interested in a game-theoretical approach to the marketing attribution problem. Relying on Game Theory models and solutions, allows the decision-maker to select a specific attribution rule based on a list of properties which she considers relevant to her problem. Moreover, she could use the proposed rules to stimulate channels to improve their efficiency through incentives. The existing literature on this topic is scarcer. Perhaps, the first paper that follows a cooperative game approach is Dalessandro et al. (2012), in which  an attribution methodology based on a causal estimation problem that uses the concept of Shapley value is proposed. Other subsequent papers adopt also the Shapley value as an attribution methodology with characteristic functions that are based on a probabilistic markovian approach. See, for instance, Singal et al. (2019).

Our approach is pure deterministic assuming that a Key Performance Index (KPI) is already defined for measuring the benefits associated with a conversion. Thus we focus on game theoretical models that can be built based on that KPI.

Some not formally published papers (Morales, 2016; Zhao et al. 2018, Cano-Berlanga) addressed this problem and, in particular, they work with the first model we analyse: the sum game. We introduce some extensions of the sum game that take into account the visiting order of the marketing channels and the possibility of repetition in the paths to conversion. The
game-theoretical solution concept we propose here as an attribution rule is the Shapley value (Shapley, 1953).

Moreover, a new different model is proposed by considering the attribution problem as a bankruptcy problem (Aumann and Maschler, 1985). For this new model, we propose the Constrained Equal Loss (CEL) and Proportional (PROP) rules as attribution rules. A similar analysis to that  did previously for the sum game model and its related Shapley attribution rule is developed in this case for the two proposed rules.

The paper is organized as follows. In section \ref{Attribution_Preli} we introduce more formally the marketing attribution problem we shall deal with.  Section \ref{ShapleyRule} introduces the general concepts of cooperative game theory we shall employ and is devoted to propose and analyse the sum game and the derived Shapley based attribution rule. In particular, we shall consider different cases depending on the relevance of the order or the repetitions of channels in the paths to conversion. A similar analysis is done for the bankruptcy model introduced in section \ref{Bank}. Some final conclusions are included in section \ref{Con}.



\section{The attribution problem}
\label{Attribution_Preli}
We are assuming that an advertising campaign exists in which an advertisement is broadcast through a set of channels. The users can have multiple touch-points with the campaign by watching the Ad in some of those channels. After this, in some moment, a conversion of a user could happen by purchasing (in a very wide sense) the advertised product producing a measurable benefit. The attribution problem is then how to attribute to the different channels that were watched before the conversion the benefit produced by that conversion. Let us formalize these ideas:

Let  $N=\{1,2,...,n\}$ be the set of channels involved in the campaign. Formally, a {\em path to conversion} is any finite ordered sequence formed with channels of $N$, $p=(i_1, i_2,...,i_{\ell_p})$, where $\ell_p$ is the length of the path $p$. We must remark that a channel can appear more than once in a path. We shall denote by $p_{(j)}\in N$, the channel that appear in the position $j$ in path $p$.

Note that the cardinal of the set of {\em all possible} paths to conversion ${\cal P}(N)$ is, in principle, infinite. However, we shall consider only finite sets of paths to conversion, since in practice only a finite numbers of paths $P(N)\subset {\cal P} (N)$ are observed. Since the benefit generated by any non realized path is zero all those paths will not belong to the support of the considered problems.

The benefit of a path is given by a Key Performance Index: $f:P(N)\longrightarrow \mathbb{R}$, that assigns to any observed path to conversion $p\in P(N)$ a measure $f(p)\geq 0$ of the benefit obtained by conversions of all users that have followed this path $p$. Thus, the total benefit of the campaign is $B=\sum_{p\in P(N)} f(p)$

In what follows, we shall assume that $f(p)$ is the sum of the benefits produced by all the users that have exactly followed the same path $p$ to the conversion. Also, we shall assume that spontaneous conversions without having watching the advertisement in any channel are already discounted in such a way that the benefit of the null path is zero.

An {\bf attribution problem} is given then by the 3-tuple $A=(N, P(N), f)$, and consists in assigning an attribution $a_i\geq 0$ to each channel $i$ in $N$ of the common benefits obtained in the campaign $B$.

For that problem some classical approaches that base the attribution in the orders the channels appears in a path to conversion, are commonly used in practice (first touch, last touch, indirect last touch, time decay, among others). We propose attribution rules based on classical rules for cooperative games with transferable utility and bankruptcy problems which allow us to get a deeper insight on the attribution mechanisms to be used in order to promote a more efficient behaviour of the channels and to help the advertiser to assign optimally her marketing budget.


	







\section{Shapley value attribution rule}
\label{ShapleyRule}

In this section we propose an attribution rule based on the Shapley value (Sahpley 1953) of an appropriate cooperative game with transferable utility, TU games in the sequel. First, let us recover some basic definitions and results regarding TU games and the Shapley value.
	
\subsection{Cooperative games}
\label{CoopGames}

A cooperative game in coalitional form with side payments, or with
 transferable utility, is an ordered pair $(N,v)$, where
 $N$ is a finite set of players and $v:2^N\rightarrow \mathbb{R}$, with $2^N=\{ S \,\vert \, S\subset N\}$, is a {\em characteristic function} on $N$ satisfying $v(\emptyset)=0$. For any coalition $S\subset N$,  $v(S)\in \mathbb{R}$ is the {\em worth} of coalition $S$ and  represents the reward that coalition $S$ can achieve by itself if all its members act together. Since we will restrict to the case of TU games in the sequel, we will refer to them simply as {\em games}. For brevity, throughout the paper, the cardinality of sets (coalitions) $N, S$ will be denoted by appropriate small letters $n,s$, respectively. Also, for notational convenience, we will write singleton $\{ i\}$ as $i$, when no ambiguity appears.

A game $(N,v)$ is {\em superadditive} if $v(S\cup T)\geq v(S)+v(T)$, for every disjoint coalitions $S\cap T\neq \emptyset$; is {\em monotone} if $v(S)\leq v(T)$, whenever $S\subseteq T$; and it is {\em convex} if $v(S\cup T)\geq v(S)+v(T)-v(S\cap T)$ for every pair $S,T\subseteq N$ of coalitions.

One of the main topics dealt with in  Cooperative Game Theory is, given a game $(N,v)\in G_n$, to divide the amount $v(N)$ between players if the grand coalition $N$ is formed. Let ${\cal U}=\{1,2,\dots\}$ be the {\em universe of players}, and let ${\cal N}$ be the class of all non-empty finite subsets of ${\cal U}$. For an  element $N$ in ${\cal N}$, let ${\cal G}_N$ denote the set of all characteristic functions on player set $N$, and let ${\cal G}=\displaystyle\bigcup_{N\in {\cal N}} {\cal G}_N$ be the set of all characteristic functions.\footnote{We will use interchangeably the two terminologies, game and characteristic function, when no ambiguity appears.}
A {\em payoff vector}, or {\em allocation}, is any $\V{x}\in \mathbb{R}^n$, which gives player $i\in N$ a payoff $x_i$.  A payoff vector is said to be {\em efficient} if $\sum_{i\in N} x_i =v (N)$. Is is {\em stable} if it is efficient and $\sum_{i\in S} x_i \geq v(S)$, for every $S\subseteq N$. The set of all stable payoff vectors is called, which will be denoted by $C(v)$ the {\em Core} of the game (Gillies 1953). The core of a game can be empty, however if the game is convex is well known that its core is nonempty.

A {\em value} $\varphi$ is an assignation which associates to each game $v\in {\cal G}_N$, $N\in {\cal N}$, a payoff vector $\varphi(N,v)=(\varphi_i(N,v))_{i\in N}\in \mathbb{R}^N$, where $\varphi_i(N,v)\in \mathbb{R}$ represents the {\em value} of player $i$, $i\in N$. Shapley (1953a) defines his value as follows:
\begin{equation}
\displaystyle
 \phi_i (N,v) = \sum_{S\subset N\setminus i} \frac{s!(n-s-1)!}{n!} \bigl( v(S\cup \{ i\}) -
v(S)
 \bigr ), \quad i\in N.
\label{shapley-marginal}
\end{equation}
 The {\em value} $\phi_i (N,v)$ of each player, which is a weighted average of his marginal contributions, can be interpreted as the {\em payoff} that player $i$ receives when the Shapley value is used to predict the allocation of resources in multiperson interactions.

A value is {\em stable} whenever  $\varphi(N,v)=(\varphi_i(N,v))_{i\in N}\in \mathbb{R}^N\in Core(N,v)$, for every game $v\in {\cal G}_N$, $N\in {\cal N}$. If the game is convex the Shapley value is stable (Shapley, 1971).

The Shapley value admits an alternative expression in terms of the {\em Harsanyi dividends} of every coalition $S$ in $(N,v)$, which are given by
\begin{equation*}
d_S=\sum_{T\subseteq S} (-1)^{|S|-|T|} v(T), \;\forall\, S\subseteq N.
\end{equation*}
The Harsanyi dividends can be calculated recursively:
\begin{equation}
    d_S=v(S)-\sum_{T\subset S} d_T, \forall S\subseteq N.
\label{HarsanyiDiv_recur}
\end{equation}
Then, the Shapley value can be expressed from the Harsanyi dividends as follows (see Shapley, 1953):
\begin{equation}
    \phi_i=\sum_{\substack{S\subseteq N \\ i\in S }} \frac{d_S}{|S|}
    \label{ShapleyDiv}
\end{equation}
The formula above is based on the linearity of the Shapley value and the expression of the game $(N,v)$ in terms of the basis of the {\em unanimity
games}\footnote{Recall that a game $(N,v)$ is a \emph{unanimity game} if there exists a coalition $S$ such that for every $T\subseteq N$, $v(T)=1$ if $S\subseteq T$, and $v(T)=0$ otherwise; in this case, we will denote the game by $(N,u^S)$ and its Shapley value is $\phi_i(N,u^S)=\frac{1}{s}$, for all $i\in S$, and  $\phi_i(N,u^S)=0$, otherwise. Unanimity games are a basis of the vector space ${\cal G}_N$.}.

A game $(N,v)$ is {\em totally positive}  when all its Harsanyi dividends are nonnegative. Every totally positive game is convex, therefore its Shapley value belongs to its Core.

%
%
%
%

\subsection{The sum game and the Shapley attribution rule}
\label{SumGame}

In this section we introduce a cooperative game-theoretical model for the attribution problem that we call the {\em sum game}. We shall consider first the case in which neither order nor repetition is considered relevant to the problem. That is, only the information of the channels in which the ad has been seen is retained. The model is then extended to the more general cases in which order or repetition may be relevant.

The main advantages of a game-theoretical approach are twofold. On the one hand, relying on Game Theory allows the decision maker to select a specific attribution rule based on a list of properties which she considers relevant to her problem. On the other hand, by the proposed rules she stimulates the channels to improve their presence and efficiency by means of incentives, since the proposed rules reward that kind of behaviour.

\subsubsection{Order non relevant case}
\label{No_Order}

We consider first the case in which neither order nor repetition is considered relevant to the attribution problem. Therefore, given the attribution problem $A=( N, P(N), f)$, the information of a given path $p$ that we retain is only the set of channels that appear in any position of the path:
\begin{equation*}
    S_p=\{i\in N\, / \, i\in p\}\subseteq N.
\end{equation*}
Thus, we can consider that we work with a set of $n$ channels $N=\{1,\dots ,n\}$ and the unique relevant information about their performance is given by the following aggregated KPI function defined over the subsets of $N$:
\begin{equation}
    f(S):=\sum_{\substack{p\in P(N) \\ S_p=S}} f(p), \;\forall \, S\subseteq N,
\label{eq:AggKPI}
\end{equation}

i.e. $f(S)$ is the total benefit produced by all users that have seen the ad exactly in the channels in $S$, regardless of the order or  the number of times.

In order to define the {\em sum game} $(N,v_{\Sigma})$ it is important to distinguish between {\em combination} of channels and {\em coalition} of channels. The first one, a {\em combination of channels} $S$ is a subset of channels in which some subset of users has seen the advertisement on all of them (and not on other channels). The second one, a {\em coalition} ${\cal S}$ is also a subset of channels, but it can form all its possible combinations with its channels. Actually, the combination $S$ and the coalition ${\cal S}$ are using the same set of channels but differ in their interpretation.

For notational convenience we will nos differentiate notation between combination and coalition. Clearly, $S$ will be a combination whenever it is an argument of the KPI function $f$ and a coalition whenever it is an argument of the characteristic function $v_{\Sigma}$

Taking into account that coalition of channels ${\cal S}$ could be also understood as the set of all possible combinations that can be formed with the channels in the coalition, the following {\em sum game} (Morales, 2016, Cano-Berlanga et al., 2017, and Zhao et al. 2018)

\begin{definition}
Given a set $N=\{1,2,...,n\}$ of players and a set function $f: 2^N \rightarrow \mathbb{R}$ such that $f(S)\geq 0$ for all $S\subseteq N$ the {\bf sum game} $(N, v_{\Sigma})$ is defined by
\begin{equation}
    v_{\Sigma}(S)=\sum_{T\subseteq N} f(T),\;  \forall \, S\subseteq N.
\label{exp_sumgame}
\end{equation}
\label{def_sumgame}
\end{definition}

It is known (Cano-Berlanga et al., 2018) that the sum game is monotone and convex, then it is superadditive, its core is nonempty and the Shapley value belongs to the core. Moreover, its  Shapley value can be simplified and expressed in terms of the KPI function $f(\cdot)$ as follows:
\begin{equation}
    \phi_i (N,v_{\Sigma})=\sum_{\substack{S\subseteq N \\ i\in S}} \frac{f(S)}{|S|},\; \forall\, i\in N.
\label{Shapley_dividends}
\end{equation}
Zhao et al. (2018) proved the above results by means of the original expression \eqref{shapley-marginal} of the Shapley value. However, we include here an alternative proof based on the Harsanyi dividends of a sum game which is more clear and allows us to achieve a deeper insight on the class of sum games.

\begin{proposition}
For any given sum game $(N, v_{\Sigma})$, the Harsanyi dividends are given by the corresponding benefit function $f(\cdot)$ that defines the game. That is, $d_S=f(S)$ for all $S\subseteq N$.
\label{dividends_sumgame}
\end{proposition}

\begin{proof}
We prove that the Harsanyi dividends of a sum game verify $d_S=f(S),$ $\forall\, S\subseteq N$,by induction on the size of $S$. Employing the recursive formula \eqref{HarsanyiDiv_recur} we obtain:

\begin{equation*}
    d_{\emptyset}=0;\,\,\,\,\,  d_i=f(i), \forall i\in N;\,\,\,\,\,  d_{\{i,j\}}=f(i,j), \forall \{i, j\}\subset N.
\end{equation*}

Let $S\subseteq N$ and assume that $d_T=f(T)$, $\forall\, T\subseteq N$ such that $t\leq s-1$. Again, by the recursive formula \eqref{HarsanyiDiv_recur}:
\begin{equation*}
    d_S=v_{\Sigma}(S)-\sum_{T\subset S} d_T=\sum_{T\subseteq N} f(T)-\sum_{T\subset S} f(T)=f(S).
\end{equation*}
\end{proof}

Note that the above expression \eqref{Shapley_dividends} follows straightforward from the previous proposition taking into account expression \eqref{ShapleyDiv} of the Shapley value in terms of the dividends of the game. Although  expression \eqref{Shapley_dividends} is simpler than the original one, it still has in principle the problem of involving a theoretically exponential number in $n$ of addends. However, in practice, the number of non-zero terms is  manageable.

As $d_S=f(S)\geq 0, \forall S\subseteq N$, the sum game is totally positive. Moreover, the sum game's class coincides with the class of totally positive games by considering that $f(S):=d_S, \forall S\subseteq N$.

In the sequel, when we refer to an {\em attribution rule} through this section we will be referring to {\em value} defined over the subclass ${\cal P}\subset {\cal G}$ of totally positive TU games.

Therefore, the classical axioms of efficiency, symmetry, null player and additivity characterize the Shapley value as an attribution rule because the sum of two totally positive games is also a totally positive game. Next, we resume these properties which characterize the Shapley value as well as some other interesting ones in terms of attribution problems.




%
\begin{itemize}
    \item  {\bf Efficiency.}
The Shapley value distributes or imputes exactly the global worth obtained from all observed combinations among the involved channels:
$$
\sum_{i\in N} \phi_i (N, v_\Sigma)=v_\Sigma ({\cal N}):=\sum_{S\subseteq N} f(S).
$$
\item{\bf Additivity.}
The Shapley value is an {\em additive} rule. Therefore,
$$
\phi_i (N, v_\Sigma^{agg})=\phi_i (N, v_\Sigma^{1}) + \phi_i (N, v_\Sigma^{2}), \; \text{for every channel $i\in N$,}
$$
where $v_\Sigma^{agg}({\cal S}):=\sum_{T\subseteq S} \bigl ( f_1(T)+f_2(T) \bigr )$, and   $v_\Sigma^{k}({\cal S}):=\sum_{T\subseteq S}  f_k(T)$, $k=1,2$, for each coalition ${\cal S}\subseteq {\cal N}$.

Starting from a dataset with a given set of channels $N$, which includes information about the KPI of two campaigns (let $f_1$ and $f_2$ be the two estimations obtained from the observed values of both campaigns), if the joint KPI of both campaigns is given by the aggregation $f_{agg} (S):=f_1(S)+f_2(S)$, for every combination $S\subseteq N$, then the Shapley value of any channel in the joint campaign equals the sum of its Shapley values in each of the original campaigns $f_1$ and $f_2$.
\item {\bf Symmetry.} For every pair $i,j\in N$ of indistinguishable channels in terms of performance, i.e. $f(S\cup i)=f(S\cup j)$ for every $S\subseteq N\setminus\{i,j\}$, the Shapley value of the sum game gives both channels the same value: $\phi_i(N,v_{\Sigma})=\phi_j(N,v_{\Sigma})$ .

\item {\bf Null player.}
The null channel property sets that a player (channel) that not make any contribution to the worth function with respect any subset of the rest of channels, i.e. $f(S\cup i)=0$, for every combination $S\subseteq N\setminus i$,
must receive a value $\phi_i(N,v_{\Sigma})=0$.

\item {\bf Stand-alone property.}
The stand-alone property, which sets that the value attributed to a channel cannot be less than the value it can obtain by itself: $\phi_i(N, v_\Sigma)\geq v(i)=f(i), \forall i\in N$.
\item {\bf Fair ranking.}
If $f(S\cup i)\geq f(S\cup j)$, for every combination $S\subseteq N\setminus \{ i,j\}$, and thus the combination with channel $i$ is more profitable (or at least equally profitable) than the combination with channel $j$ for every combination $S$, channel $i$ should rank better than channel $j$. That is, $\phi_i(N, v_\Sigma)\geq \phi_j(N,v_\Sigma)$.

\item {\bf Stability.}
When using the Shapley value the global value imputed to each combination $S$ of channels is always greater or equal than the value generated by all conversions from users exposed to any possible combination of channels in $S$ that have not received impressions in any other channel  (which is precisely given by $v_\Sigma({\cal S})$). Thus, the combination of channels in $S$ has no incentive to reject the proposed attribution scheme.
$$
\phi(S)(N,v_\Sigma):=\sum_{i \in S} \phi_i (N,v_\Sigma) \geq v_\Sigma({\cal S})=\sum_{ T\subseteq S}f(T),\quad \forall\, S\subseteq N .
$$
\item {\bf No subsidizing property.}
First, let us formally define the concept of {\bf independent set of channels}: a given subset $S^*\subseteq N$ is an {\em independent set of channels} if $f(S\cup T)=0$, for all $S\subseteq S^*$ and all $\emptyset \neq T\subseteq N\setminus S^*$.

Note that the worth generated by an independent subset of channels equals $\sum_{S\subseteq S^*} f(S)$, which is precisely $v_\Sigma({\cal S}^*)$, and moreover, it can be clearly identifiable and thus should be imputed to $S^*$. In these situations, in which no conversion (or at least a worthy conversion) has been made by users exposed to combinations of channels mixing some channels of $S^*$ with some channels not in $S^*$, channels in $S^*$ should not receive any credit from conversions of devices exposed to channels in $N\setminus S^*$, and the other way around.

This is the idea of the following property.

If there exists and independent subset of channels $S^*\subseteq N$, then the Shapley value imputes to those channels in $S$ exactly the global value generated by them, which in that case coincides the value generated by all conversions from users exposed to any possible combination of channels in $S^*$ that have not received impressions in any other channel.
$$
\phi(S^*):=\sum_{i \in S^*} \phi_i (N,v_\Sigma) = v_\Sigma({\cal S}^*)=\sum_{ T\subseteq S^*}f(T), \quad \forall \text{ independent }  S^*\subseteq N.
$$

\end{itemize}



%
%
%
%



Next, we enrich the sum game approach in order to deal with more general attribution situations in which the order in which channels appear in the conversion path as well as the number of times each channel appears play a relevant role. First, we consider in section \ref{No_OrderRep} the case in which only the number of times a channel appears in a conversion path is relevant. Then, in section \ref{Order} we take into account the order of appearance.

\subsubsection{Considering repetitions}
\label{No_OrderRep}

In this case, the order in not considered relevant, but it is considered relevant the fact that a channel can appear more than once in a path.

\begin{example}
The paths $(1, 2)$ and $(2, 1)$ are considered not distinguishable, but different from the path $(2, 1, 2, 2)$, for instance. In this case, we consider that the worth of channel $2$ must be greater than the worth of channel $1$ in that path to conversion.
\label{Ex1}
\end{example}

The above approach of aggregating the values of different paths that share the same subset $S$ of players for obtaining the value $f(S)$ and later the value $v_{\Sigma}(S)$ will be done with precaution now.
In the previous example, if we define $f(\{1, 2\})=f((1,2))+f((2,1))+f((2, 1, 2, 2))$ the information about the repetition of Ad views in the channel $2$ in the third path is lost. Then, we propose to create additional players that replicates a player that appear more than once in a path.

Formally, let $r_i$ the maximum number of times player $i$ appears in any path of the attribution problem $( N,P(N),f )$, i.e.
\begin{equation}
r_i=\max_{p\in P_i(N)} n_i(p),
\label{def:ri}
\end{equation}
where $n_i(p)$ is the number of times channel $i$ appears in path $p$, $i=1,\dots,n$. Then we create fictitious players $i^1, i^2,\dots,i^{r_i}$ that substitute original channel $i$. To be specific, if channel $i$ appears $\ell=n_i(p)$ times in path $p$, then channel $i$  is substituted by fictitious players  $i^1, i^2,\dots,i^{\ell}$ in this path $p$, for all $p\in P_i(N)$, and for all $i\in N$.  For those new  players the sum game $(N^r, v^r_{\Sigma})$ and its corresponding Shapley value are defined accordingly to  previous section's definitions taking into account that $N^r=\cup_{i=1}^n S^r_i$, where $S^r_i =\{i^1, i^2,\dots,i^{r_i}\}$, $i=1,\dots,n$, and being the KPI function $f^r$ -defined over combinations of fictitious channels
 in $N^r$- given by the following sum:
\begin{equation}
f^r(S^r):=\sum_{\substack{p\in P(N) \\ n_j(p)=|S^r\cap S_j^r|, \forall j}} f(p), \; \forall \, S^r\subseteq N^r.
\label{def:KPIrep}
\end{equation}

\begin{example} Let us consider the case described above in Example \ref{Ex1} with the following KPI values:

\begin{table}[h]
\caption{Example \ref{Ex1}  with repetitions}
\label{DataRepEx}       
\begin{tabular}{lc}
\hline\noalign{\smallskip}
Path $p$ & KPI value $f(p)$
\\
\noalign{\smallskip}\hline\noalign{\smallskip}
$(1)$ & 20  \\
$(1,2)$ & 40  \\
$(2,1)$ & 10  \\
$(2,1,2)$ & 30  \\
\noalign{\smallskip}\hline
\end{tabular}
\end{table}

Then we create fictitious players $1^1, 2^1, 2^2$ and consider the sum game $(N^r=\{1^1,2^1,2^2\},v_{\Sigma}^r)$ based on the KPI function $f^r$ depicted in the new table \ref{DataRepKv}:
%
%

\begin{table}[h]
\caption{KPI and characteristic function for combinations and coalitions}
\label{DataRepKv}       
\begin{tabular}{lcc}
\hline\noalign{\smallskip}
$S^r$ & $f^r(S)$ & $v^r_{\Sigma}(S)$  \\
\noalign{\smallskip}\hline\noalign{\smallskip}
$\{ 1^1\} $ & 20  & 20\\
$\{2^1\}$ & 0 & 0  \\
$\{2^2\}$ & 0 & 0  \\
$\{1^1, 2^1\}$ & 50 & 70  \\
$\{1^1,2^2\}$ & 0 & 20 \\
$\{2^1,2^2\}$ & 0 & 0  \\
$\{1^1,2^1,2^2\}$ & 30 & 100 \\
\noalign{\smallskip}\hline
\end{tabular}
\end{table}

Thus, the Shapley value of the sum game $(N^r,v^r_{\Sigma})$ are given by $\phi_{1^1}(N^r,v^r_{\Sigma}) =55$, $\phi_{2^1}(N^r,v^r_{\Sigma})=35$, and            $\phi_{2^2}(N^r,v^r_{\Sigma})=10$.

Now, taking into account that the Shapley value of each channel when nor order neither repetition are relevant\footnote{In this case, the value of coalition $S=\{1,2\}$ will be given by the sum $f(1)+f(1,2)+f(2,1)+f(2,1,2)=100$.} is $\phi_1(N,v_{\Sigma})=60$ and $\phi_2(N,v_{\Sigma})=40$, we can observe that $\phi_{2^1}(N^r,v^r_{\Sigma})+\phi_{2^2}(N^r,v^r_{\Sigma})=45>\phi_2(N,v_{\Sigma})$. That is, channel $2$ increases its worth due to the repetition in one path.
\end{example}

Next, we formalize the ideas shown in this example. First we define an attribution when repetition is relevant, which we refer to as {\em Shapley value-like attribution} that is based on the Shapley value of the extended sum game $(N^r,v^r_{\Sigma})$. Then, we deduce a simple expression for it based on the KPI of the observed paths, and we introduce the property of {\em monotonicity with respect to channel repetition} to capture the fact that if a given channel increases its appearances while the remaining characteristics stay its attribution will improve, or at least will not worsen. We end up this case showing that the proposed Sahpley value-like attribution preserves the essential properties of the Shaple value:
efficiency, additivity, symmetry and null player.

\begin{definition}
Formally, the {\bf Shapley value-like attribution} of each channel $i\in N$ in this framework will be given by the following sum:
\begin{equation}
   \phi_i^r(N,P(N),f):= \phi_{i^1} (N^r,v^r_{\Sigma})+\cdots +\phi_{i^{r_i}}(N^r,v^r_{\Sigma}), \; i\in N,
\label{def:ShapleyRep_eq}
\end{equation}
where $S_i^r=\{i^1, i^2,...,i^{r_i}\}$ and being $r_i\geq 1$ the maximum number defined in \eqref{def:ri}, for every channel $i\in N$.
\label{def:ShapleyRep}
\end{definition}

\begin{proposition}
Let $A=( N, P(N), f)$ be an attribution problem, and let $r_i\geq 1$ be the maximum number defined in \eqref{def:ri}, for every channel $i\in N$. Then, it holds:
\begin{equation}
   \phi_i^r(N,P(N),f)=\sum_{\substack{S^r\subseteq N^r \\ S^r\cap S_i^r\neq \emptyset}}  \frac{|S^r\cap S_i^r|}{|S^r|} f^r(S^r)=
    \sum_{\substack{p\in P(N) \\ i\in p}}  \frac{n_i(p)}{\ell_p} f(p),
\label{eq:ShapleyRep}
\end{equation}
where $S_i^r=\{i^1, i^2,...,i^{r_i}\}$, $n_i(p)$ is the number of times channel $i$ appears in path $p$, being $\ell_p$ its length, $i=1,\dots,n$.
\end{proposition}

\begin{proof}
Since $\phi_i^r(N,P(N),f)$ is by definition the sum:
$$
 \phi_{i^1} (N^r,v^r_{\Sigma})+\cdots +\phi_{i^{r_i}}(N^r,v^r_{\Sigma}),
$$
for all $i\in N$, the first equality follows straightforward from expression \eqref{Shapley_dividends} of the Shapley value of a sum game in terms of its Harsanyi dividends and the extended game $(N^r,v^r_{\Sigma})$ definition.

Second equality follows from the definition of the extended channel set $N^r$ and KPI function $f^r(\cdot)$ expression \eqref{def:KPIrep}.
\end{proof}

Le us now define the property of {\em monotonicity with respect to channel repetition} to describe those rules verifying that {\it ceteris paribus}   the repetition of a channel favours it.

\begin{definition}
Let $A=(N,P(N),f)$ and $A^{+i}=(N,P^{+i}(N), f^{+i})$ be two attribution problems such that there exists a path $p\in P_i(N)$ with $P^{+i}(N)=P(N)\setminus \{ p\} \cup p^{+i}$, where $p^{+i}$ substitutes the path $p$ by repeating once channel $i\in p$, but without changing its value, i.e.,  $f^{+i}(p^{+i})=f(p)$, and being $f^{+i} (q)=f(q)$ for every  $q\in P(N)\setminus \{ p\}$. Then an attribution rule $\psi$ verifies  {\bf monotonicity with respect to channel repetition} whenever $\psi_i(N,P(N),f)\leq \psi_i(N,P^{+i}(N), f^{+i})$.
\end{definition}

\begin{proposition}
The Shapley value-like attribution rule introduced in definition \ref{def:ShapleyRep} verifies monotonicity with respect to channel repetition.
\end{proposition}

\begin{proof}
Taking into account the last expression in \eqref{eq:ShapleyRep}, the unique difference between $\phi_i(N,P(N),f)$ and  $\phi_i(N,P^{+i}(N), f^{+i})$ is given by the weight corresponding to paths $p\in P(N)$ and $p^{+i}\in P^{+i}(N)$ which are
 $\frac{n_i(p)}{\ell_p}$ and $\frac{n_i(p)+1}{\ell_p+1}$, respectively. Since  $\ell_p\geq n_i(p)$, the inequality holds.
\end{proof}

We end up the analysis of this case by proving that the proposed attribution is a Shapley-like attribution in the sense that it preserves the essential properties  characterizing the Shapley value: efficiency, additivity, symmetry and null player. Previously, we must clarify the meaning of {\em symmetric} channels in this generalized context.

Let $A=(N,P(N),f)$ be an attribution problem with repetitions, then two different channels $i\neq j\in N$ are {\em symmetric} if $r_i=r_j$ and $f((p,\overbrace{i,\dots,i}^{r}))=f((p,\overbrace{j,\dots,j}^{r}))$ for all path $p$ and $1\leq r\leq r_i$. Note that in principle this equality must hold for an infinite number of paths. However, only a finite amount of conditions regarding observed paths are not trivial identities $0=0$. In fact, $i,j$ symmetric implies that  $(p,{\overbrace{i,\dots,i}^{r})}\in P(N) \, \Leftrightarrow \, (p,\overbrace{j,\dots,j}^{r})\in P(N)$.

\begin{proposition}
The Shapley value-like attribution verifies the properties of efficiency, additivity, symmetry and null player.
\label{th:properties_ShapleyLike}
\end{proposition}

\begin{proof}
Efficiency follows from the definition of the extended game $(N^r,v^r_{\Sigma})$ and efficiency of the Shapley value, since:
$$
\sum_{i=1}^n \phi_i(N,P(N),f)= \sum_{i=1}^n \sum_{\ell=1}^{r_i}\phi_{i^{\ell}}(N^r,v^r_{\Sigma})=v^r_{\Sigma} (N^r)=\sum_{S^r\subseteq N^r} f^r(S^r)=B,
$$
being $B$ the global benefit of the campaign.

Symmetry and null player can be proved following a similar reasoning.

In order to prove additivity, let $(N,P^1(N),f_1)$ and $(N,P^2(N),f_2)$ be two campaigns developed for the same set of channels $N$, then we must first consider a new extended player set $N_{12}^r:=N_{1}^r\cup N_{2}^r$ and consider $r^{1,2}_i:=\max \{r_i^1,r_i^2\}$ to account for all possible repetitions in both campaigns. Clearly, if $r_i^1<r_i^2$ then all paths $p$ in which channel $i$ appears $\ell > r_i^1$ tiems shall have a zero KPI in the first campaign, and therefore the corresponding fictitious players $i^{\ell}$ with $r_i^1<\ell\leq r_i^2$ will be {\em null players} in $(N^{r}_{12},v^r_{1,\Sigma})$. However, since the Shapley value verifies  {\em Null player out} property\footnote{Removing a null player does not affect the Shapley value of the remaining players: $\phi_i(N,v)=\phi_i(N\setminus j, v_{-j})$, for all $i,j\in N$, $v\in {\cal G}_N$, such that $j$ is a null-player in $(N,v)$ and $i\neq j$.} (Derks and Haller, 1999), it holds for all $i\in N$:
\begin{align*}
   \phi_{i^{\ell}}(N^{r}_{12},v^r_{1,\Sigma}) & = \phi_{i^{\ell}}(N^{r}_{1},v^r_{1,\Sigma}), \; \forall \, \ell \leq r_i^1,
   \\
   \phi_{i^{\ell}}(N^{r}_{12},v^r_{1,\Sigma}) &  = 0, \; \forall \, r_i^1 < \ell\leq  r_i^{1,2},
  \\
    \phi_{i^{\ell}}(N^{r}_{12},v^r_{2,\Sigma}) & = \phi_{i^{\ell}}(N^{r}_{2},v^r_{2,\Sigma}), \; \forall \, \ell \leq r_i^2,
    \\
   \phi_{i^{\ell}}(N^{r}_{12},v^r_{2,\Sigma})  & = 0, \; \forall \, r_i^2 < \ell\leq  r_i^{1,2}.
\end{align*}
Thus, by definition of Shapley value-like attribution and additivity of the Shapley value, it follows that $\phi_i(N,P^1(N), f^1) + \phi_i(N,P^2(N), f^2)$ equals:
\begin{align*}
  \displaystyle \sum_{\ell=1}^{r_i^1} \phi_{i^{\ell}}(N^{r}_{1},v^r_{1,\Sigma}) +  \sum_{\ell=1}^{r_i^2} \phi_{i^{\ell}}(N^{r}_{2},v^r_{2,\Sigma})  =
    \displaystyle  \sum_{\ell=1}^{r_i^{1,2}} \phi_{i^{\ell}}(N^{r}_{12},v^r_{1,\Sigma}) +  \sum_{\ell=1}^{r_i^{1,2}} \phi_{i^{\ell}}(N^{r}_{12},v^r_{2,\Sigma})  =
  \\[3mm]
 \displaystyle
 \sum_{\ell=1}^{r_i^{1,2}} \phi_{i^{\ell}}(N^{r}_{12},v^r_{1,\Sigma}+v^r_{2,\Sigma}):=\phi_i(N,P_1(N)\cup P_2(N), f^1+f^2),
\end{align*}
for every channel $i\in N$, where $(N,P_1(N)\cup P_2(N), f^1+f^2)$ is precisely the attribution problem describing the combination of campaigns 1 and 2.
\end{proof}

\subsubsection{Order relevant case}
\label{Order}

We consider now the case in which a channel can play different roles at different stages of the conversion process, i.e. a channel can have different impacts on users' decision making at different stages of the path to conversion, so it would be very useful if we could understand the role that each channel can play in each step of the conversion process and evaluate correctly the attribution values taking into account the order in which it appears in each path to conversion.


For taking into account the order, again we make use of the idea of artificially replicating players. In this case, every channel $i$ that appears in position $j$ in some path to conversion $p\in P(N)$ is substituted by a new fictitious channel $i_j$ that combines the information about channel and position. For this new set of extended order players $N^o$, the order is again considered irrelevant and the KPI function and the extended order sum game $(N^o, v^o_{\Sigma})$ are defined in the same way as in section \ref{No_Order}. Now, the extended order set is given by $N^o=\cup_{i=1}^n S^o_i$, where $S^o_i=\{ i_1,\dots, i_{p_i}\}$, being $p_i$ the maximum position that channel $i\in N$ reaches in the set of paths to conversion to which it belongs.

\begin{example}
Let us consider the attribution problem described by the data in Table \ref{DataOrder}.
%
\begin{table}[h]
\caption{Data for order case}
\label{DataOrder}       
\begin{tabular}{lc}
\hline\noalign{\smallskip}
Path $p\in P(N)$ & KPI value $f(p)$  \\
\noalign{\smallskip}\hline\noalign{\smallskip}
$(1)$ & 30  \\
$(1,2)$ & 60  \\
$(2, 1)$ & 10  \\
\noalign{\smallskip}\hline
\end{tabular}
\end{table}

%
The data for the extended order players, the corresponding values for combinations and coalitions are depicted in Table \ref{DataOrderKv}. Note that $N^o=\{1_1,1_2,2_1,2_2\}$ and not all possible combinations and coalitions are considered. $f^o(S)=0$ for every combination $S$ not included in Table \ref{DataOrderKv}.
\begin{table}[h]
\caption{KPI and characteristic functions for the extended ordered set $N^o$}
\label{DataOrderKv}       
\begin{tabular}{lcc}
\hline\noalign{\smallskip}
$S^o\in N^o$ & $f^o(S)$ & $v^o_{\Sigma}(S)$  \\
\noalign{\smallskip}\hline\noalign{\smallskip}
$\{ 1_1 \}$ & 30  & 30\\
$\{ 2_1\}$ & 0 & 0  \\
$\{ 2_2\}$ & 0 & 0  \\
$\{ 1_1, 2_1\}$ & 0 & 30  \\
$\{ 1_1,2_2\}$ & 60 & 90 \\
$\{ 1_2, 2_1\}$ & 10 & 10  \\
$\{ 1_1,2_1,2_2\}$ & 0 & 90 \\
$\{ 1_1,1_2,2_1,2_2\}$ & 0 & 100 \\
\noalign{\smallskip}\hline
\end{tabular}
\end{table}

The Shapley value of the extended order game $(N^o,v^0_{\Sigma})$ is given by $\phi_{1_1}=60$, $\phi_{1_2}=5$, $\phi_{2_1}=5$, and  $\phi_{2_2}=30$. We can observe the following relation with the Shapley value of the sum game when the order is not considered relevant:
\begin{eqnarray}
\phi_{1_1} (N^o,v^0_{\Sigma})+\phi_{1_2} (N^o,v^0_{\Sigma})=65=\phi_1 (N,v_{\Sigma}), \label{relation_order_non1}
\\
\phi_{2_1} (N^o,v^0_{\Sigma})+\phi_{2_2} (N^o,v^0_{\Sigma})=35=\phi_2(N,v_{\Sigma}).
  \label{relation_order_non2}
\end{eqnarray}
Then, we can interpret the attribution to channels $1$ and $2$ as the sum of the attribution obtained by each channel when it occupies the first position or the second position in a path. In this particular case, channel $1$ contributes much more when it is the first touch-point to conversion, whereas channel $2$ contributes much more when it is the last touch-point.
\end{example}

For convenience, in what follows we shall denote the Shapley value of channel $i$ in position $j$ in the extended order game $\phi_{i_j}(N^o,v^o_{\Sigma})$ by $\phi_i^j (N^o,v^o_{\Sigma})$, which can be obtained by means of the following simplified expression \eqref{shapley_order_position}.

\begin{proposition}
For any attribution problem $A=( N, P(N),f)$  it holds:
\begin{equation}
    \phi_i^j(N^o,v^o_{\Sigma})=\sum_{p\in P_i^j(N)} \frac{f(p)}{\ell_p}
    \label{shapley_order_position}
\end{equation}
where $P_i^j(N)\subseteq P(N)$ is the set of paths in which player $i$ occupies position $j$.
\label{prop:shapley_order_position}
\end{proposition}

\begin{proof}
Follows straightforward from expression \eqref{ShapleyDiv} of the Shapley value of a sum game in terms of the KPI function taking into account that for all $S^o\in N^o$ such that $i_j\in S^o$ it exists a path $p'\in P(N)$ in which player $i$ occupies position $j$. Thus, for all $i_j\in N^o$ holds:
$$
\phi_i^j(N^o,v^o_{\Sigma}) =  \sum_{\substack{ S^o \subseteq N^o \\ i_j\in S^o}} \frac{f^o(S^o)}{\vert S^o\vert}=\sum_{p\in P_i^j(N)} \frac{f(p)}{\ell_p} .
$$
\end{proof}

If no repetition occurs and thus each channel appears only once in every observed path to conversion $p\in P(N)$, next proposition shows that relations \eqref{relation_order_non1} and \eqref{relation_order_non2} generalize. Then, we can always interpret the Shapley attribution to each channel when ignoring positions as the sum of the different Shapley attributions obtained by that channel when it occupies different positions in a path.

\begin{proposition}
Let $A=( N, P(N),f)$ be an attribution problem such that every channel $i\in N$ appears at most one time in each path $p\in P(N)$. Then, for any $i\in N$, $\phi_i(N,v_{\Sigma})=\sum_{j=1}^{p_i} \phi_i^j (N^o,v^o_{\Sigma})$, where $p_i$ is the maximum position that $i$ reaches in the set of paths it belongs, $P_i\subseteq P(N)$.
\label{decom_order_Shapley}
\end{proposition}

\begin{proof}
Since there are repetitions in any observed path to conversion, $\vert S_p\vert =\ell_p$, for all $p\in P(N)$, and therefore:
$$
\phi_i(N,v_{\Sigma})= \sum_{\substack{ S \subseteq N \\ i\in S}} \frac{f(S)}{\vert S\vert}=
 \sum_{p\in P_i(N)} \frac{f(p)}{\ell_p} =\sum_{j=1}^{p_i}  \sum_{p\in P_i^j(N)} \frac{f(p)}{\ell_p}.
$$
Thus, taking into account proposition \ref{prop:shapley_order_position} the decomposition result hods.
\end{proof}

 For the more general case, in which a channel appears more than one time in some paths, and repetition is also relevant, the approach of the previous section can be employed in combination with the one in this chapter, and we can define again a Shapley-value like attribution per each channel. However, the decomposition theorem \ref{decom_order_Shapley} is no longer true since also repetition effect over the attribution of each channel is measured.
	
	
It must be remarked that in Zhao et al. (2018) an intuitive idea about the use of a similar approach to deal with the case in which the order is relevant is considered, which they called ''ordered Shapley values''. However, they did not consider any formalization of the procedure to obtain these ordered Shapley values or its relation to the case that the order was not relevant.

Following Zhao et al. (2018) we can measure the importance of a given fixed position $j$ by means of the sum of the Shapley values of the ordered players $i_j\in N^o$ corresponding to this position. Formally:

\begin{definition}
For any attribution problem $A=( N, P(N),f)$, and any observed position $j$ (i.e., there exists a path $p\in P(N)$ of length $\ell\geq j$), the {\bf Shapley value-like attribution} of position $j$ is defined as:
\begin{equation*}
    \phi^j(N,P(N),f):=\sum_{i\in N} \phi_i^j (N^o,v^o_{\Sigma}).
\end{equation*}
\end{definition}

The Shapley value-like contribution of position $j$ defined above can be also obtained in terms of the KPI function as follows.

\begin{proposition}
For any attribution problem $A=( N, P(N),f)$, and any observed position $j$, it holds:
\begin{equation*}
 \phi^j(N,P(N),f)=\sum_{\substack{p\in P(N) \\ \ell_p\geq j}} \frac{f(p)}{\ell_p},
\end{equation*}
\end{proposition}

\begin{proof}
Straightforward since $\displaystyle \bigcup_{i\in N} P_i^j(N)=\{ p\in P(N)\, /\, \ell_p\geq j\}$.
\end{proof}

Clearly, the indexes $\phi_i^j$, $j=1,\dots,p_i$, $i=1,\dots, n$ for sharing the benefits attributed to channel $i$ among the different positions it can occupy are Shapley value allocations and therefore trivially satisfy all properties listed in section \ref{No_Order}. But also the indexes $\phi^j$, for sharing the benefits produced among the different positions are Shapley-like allocations in the sense that they also verify the basic properties of the Sahpley value.

\section{A bankruptcy approach}
\label{Bank}

In this section we propose a different game theoretical approach to the attribution problem considering it as a bankruptcy problem. The bankruptcy problem was introduced as a game theoretical problem by Aumann and Maschler (1985) for solving the problem of how to allocate a given state among the different agents that have rights on part of it in the case in which the estate is not sufficient to meet all their claims.

Formally, the bankruptcy problem is given by $(E, c)$, where $E$ is the {\em estate} and $c=(c_1,\dots, c_n)$ is the {\em vector of claims}, being $c_i$ the claim of agent $i$, such that $0<E\leq \sum_{i=1}^{n} c_i$.
Let ${\cal U}=\{1,2,\dots\}$ be the {\em universe of claimants}, and let ${\cal N}$ be the class of all non-empty finite subsets of ${\cal U}$. For an  element $N$ in ${\cal N}$, let ${\cal B}_N$ denote the family of all those bankruptcy problems defined on $N$, and let ${\cal B}=\displaystyle\bigcup_{N\in {\cal N}} {\cal B}_N$ be the set of all bankruptcy problems.

Let $B=(E,c)\in {\cal B}_N$ be a given bankruptcy problem. We shall denote by $C=\sum_{i=1}^{n} c_i$ the total quantity that is claimed and by $D=C-E\geq 0$ the deficit. Then, a cooperative game $(N, v)$ can be associated to this problem with characteristic function:
\begin{equation}
v(S)=\max\{0, E-\sum_{i\notin { S}} c_i\},\quad \, S\subseteq N.
\label{bankruptcy-pesimistic-game}
\end{equation}

It is well known that this game is convex. Although the Shapley value for this game could be calculated, it does not have good properties as a bankruptcy solution. Instead, usually some other specific bankruptcy rules are employed.

A {\bf rule} for ${\cal B}^N$ is a mapping $R$ that associates with every problem $(E, c)\in {\cal B}^N$ a unique value $R(E,c)\in \mathbb{R}^n$.\footnote{Sometimes the set of involved players could vary, so that we must explicitly include it in the notation by writing $(N,E,c)$ instead of $(E,c)$, for instance.}

In particular, for the attribution problem we shall pay attention in this paper to the {\em Proportional} (PROP) and the {\em Constrained Equal Losses} (CEL) rules (Aumann and Maschler 1985). The PROP rule, which distributes the estate proportionally to the claims and is probably the most widely used solution rule in this framework, can provide us with a benchmark. The CEL rule, which has the property of excluding weaker claimants helps us to concentrate on the must powerful channels and excluding from the sharing-out those channels with low contributions.

\begin{definition}
The {\bf proportional rule} assigns to each agent a part of the estate proportional to its claim:
\[
PROP_i=\frac{c_i}{C}E, \quad i=1,\dots,n,
\]
where $C=\sum_{i\in N} c_i$, for every bankruptcy problem $B=(E,c)\in {\cal B}_N$.
\end{definition}

The {\bf CEL rule} follows an approach that tries to impute equally the deficit $D$ to the claimants:

\begin{definition}
The {\bf CEL (Constrained Equal Losses rule)} assigns to each agent in a bankruptcy problem $B=(E,c)\in {\cal B}_N$ the amount:
\[
CEL_i=\max \{0, c_i-\lambda\},
\]
\noindent where $\lambda >0$ verifies
\[
\sum_{i\in N} \max \{0, c_i-\lambda\}=E.
\]
\end{definition}


\subsection{The attribution problem as a bankruptcy problem}

We shall consider now that each channel claims (all) the benefits generated by all those conversions related to users that have seen the ad on that channel.

In the first place, we will consider the case in which we record only if the ad has been seen on a channel or not, without taking into account the order or the number of times it has been seen; afterwards, the relevant case of order and repetition will be considered.

\subsubsection{No order, no repetition}

Therefore, given an attribution problem $( N, P(N), f) $ we shall consider that the only relevant information is the KPI function $f$ defined on the combinations $S$ of $N$. We denote by $(F, c)$ the associated bankruptcy problem, where the estate is given by
\[
F=\sum_{S\in N)} f(S),
\]
\noindent i.e. is the total value produced by the set of channels $N$ in a particular campaign, and the individual claims are given by
\[
c_i=\sum_{\substack{S\in N\\  i\in S}} f(S),\quad i\in N,
\]
That is, each channel claims (all) the benefits generated by all those conversions of users that have seen the ad on that channel. Obviously, the sum of the claims  $C=\sum_{i \in N} c_i$ exceeds the global benefit $F$ (the estate), which is precisely the global amount that must be attributed to the channels.

We shall denote by $D\geq 0$ the deficit $D=C-F$. From the definition of $D$ we can deduce that:
\[
D=\sum_{\substack{S\subset N \\ S\neq \emptyset}}(|S|-1)f(S).
\]

The family of all those bankruptcy problems involving channels in $N$ is denoted by ${\cal F}^N$. As we will see in the next theorem ${\cal F}^N$ is a proper subset of ${\cal B}^N$, for every finite element $N$ in ${\cal N}$. Let ${\cal F}=\cup_{N\subseteq {\cal N}} {\cal F}^N$.

\begin{theorem}
For every $N\in {\cal N}$ and any given bankruptcy problem $B=(E, c)\in {\cal B}^N$,
\begin{equation}
\sum_{i \in N} c_i\geq E\geq \max_i c_i
\label{attr_cond}
\end{equation}
is a necessary and sufficient condition for the existence of a non negative set function $f:2^N\rightarrow \mathbb{R}_+$ such that
$$
E=\sum_{\substack{S\subset N \\ S\neq \emptyset}}f(S) \text{ and } c_i=\sum_{\substack{S\subseteq N\\  i\in S}} f(S), \; \forall \, i\in N.
$$
\end{theorem}

\begin{proof}
The necessity of the condition is straightforward. Sufficiency is proving by induction over the number of claimants.

Let us first prove that, under those conditions, the next set of linear constraints corresponding to a bankruptcy problem with two claimants, has a non-negative feasible solution.
\begin{align*}
E &= f(1)+f(2)+f(\{ 1,2\}),
\\
c_1 &= f(1) + f(\{ 1,2\}),
\\
c_2 &= f(2) + f(\{ 1,2\}).
\end{align*}
Trivially, $f(1)=E-c_2$, $f(2)=E-c_1$ and $f(\{ 1,2\})=c_1+c_2-E$ solves the system and condition \eqref{attr_cond} assures that all of them are non-negative.

Now, let us suppose by induction hypothesis that given a bankruptcy problem with $n=\vert N\vert$ claimants satisfying condition \eqref{attr_cond} there exists a non-negative set function $f(\cdot)$ such that $E=\sum_{\substack{S\subset N \\ S\neq \emptyset}}f(S)$   and  $c_i=\sum_{\substack{S\subseteq N\\  i\in S}} f(S)$, for all $i\in N$, and we will prove the existence of a similar function for any bankruptcy problem with $n+1$ claimants verifying \eqref{attr_cond}.

Let $B=(E,c)$ be such a bankruptcy problem, and let us assume without loss of generality $c_1\leq c_2\leq \cdots \leq c_n\leq c_{n+1}$. Then let us define
\begin{equation}
f(n+1):=\max \{ c_{n+1}-c_n, E-\displaystyle \sum_{i=1}^n c_i\}
\label{fmaxdef}
\end{equation}
Now, let us consider two cases:
\begin{enumerate}
\item If $f(n+1)=c_{n+1}-c_n$, then the reduced problem $B':=(E',(c_1,\dots,c_n))$ with $E':=E- c_{n+1}+c_n$ is a bankruptcy problem satisfying condition \eqref{attr_cond}. Note that $c_n=\max_{i=1,\dots,n} c_i \geq E-c_{n+1} + c_ n =E'$ since $E\geq c_{n+1}$, and $\displaystyle \sum_{i=1}^n c_i \geq E - c_{n+1} + c_n=E'$, since $c_{n+1}-c_n \geq E-\displaystyle \sum_{i=1}^n c_i$. Thus, there exists a non-negative function $f'(\cdot)$ such that:
\begin{align}
E -  c_{n+1} + c_n &= \displaystyle \sum_{S\subseteq N} f'(S),
\\
c_i &= \displaystyle \sum_{\substack{S\subseteq N \\ i\in S}} f'(S), \; i=1,\dots, n.
\label{equationBn}
\end{align}
Now, let us define function $f: 2^{N\cup \{n+1\}}\rightarrow \mathbb{R}_+$ as follows: $f(n+1):= c_{n+1}-c_n$; $f(S):=f'(S)$ and $f(S\cup \{ n,n+1\}):=f'(S\cup \{n\})$, for every $S\subseteq \{1,\dots, n-1\}$; and being $f(S):=0$ for the remaining $S$ containing only one of the agents $n$ or $n+1$.

Trivially, $f$ verifies
$$
E=\sum_{\substack{S\subset N \cup \{n+1\}\\ S\neq \emptyset}}f(S)
$$
and
$$
c_i=\sum_{\substack{S\subseteq N\cup\{ n+1\}\\  i\in S}} f(S), \text{ for all $i\in N\cup \{ n+1\}$}.
$$
\item
If $f(n+1)=E-\displaystyle \sum_{i=1}^n c_i$, then $0\leq c_{n+1}-c_n \leq E-\displaystyle \sum_{i=1}^n c_i$, and therefore the reduced problem $B':=(E',(c_1,\dots,c_n))$ with $E':=E- f(n+1)=\sum_{i\in N} c_i$ is a trivial bankruptcy problem which satisfies also condition \eqref{attr_cond}.

The function $f'(i):=c_i$ for all $i=1,\dots,n$, and $f'(S)=0$ otherwise, is a non-negative function that trivially verifies
$$
E':=\sum_{i=1}^n c_i=\sum_{\substack{S\subset N \\ S\neq \emptyset}}f'(S) \text{ and  } c_i=\sum_{\substack{S\subseteq N\\  i\in S}} f'(S), \; \forall \, i\in N=\{1,\dots,n\}.
$$
 Now, let us define the function $f$ as follows:
\begin{align*}
  f(n+1) & := E-\displaystyle \sum_{i=1}^n c_i,
  \\
  f(i) & :=f'(i)=c_i, \text{ $i=1,\dots,n-1$},
   \\
   f(n) & :=f'(n)-D=c_n-(\sum_{i=1}^{n+1}-E)\geq 0 \text{ (since $c_{n+1}-c_n\leq =E-\displaystyle \sum_{i=1}^n c_i$)},
   \\
   f(\{n,n+1\}) & :=D=\sum_{i=1}^{n+1}-E\geq 0,
   \\
   f(S) & :=0,  \text{ otherwise.}
   \end{align*}
 Trivially, $f$ verifies
 $$
 E=\sum_{\substack{S\subset N \cup \{n+1\}\\ S\neq \emptyset}}f(S) \text{ and  } c_i=\sum_{\substack{S\subseteq N\cup\{ n+1\}\\  i\in S}} f(S), \; \forall \, i\in N\cup \{ n+1\}.
 $$
\end{enumerate}

\end{proof}

In the sequel, let $A=(N,P(N),f)$ be an attribution problem, when we refer to its {\it Constrained Equal Loss (CEL)} or {\it Poportional (PROP)} attribution we were referring to the CEL and PROP, respectively, shares of the corresponding bankruptcy problem $B=(N,F,c)$ associated to $A$.

Since ${\cal F}^N$ is a proper subset of ${\cal B}^N$, for every finite element $N$ in ${\cal N}$, the classic axiomatic characterizations of Proportional and CEL rules are not valid when restricted to this problem. New characterizations must be developed for this subclass of bankruptcy problems. In fact, some of the key properties that characterize those solutions, such {\em consistency}, for instance, have no meaning in this framework. Let us analyze the classical properties for bankruptcy rules which are relevant\footnote{They appear in some characterization of these bankruptcy rules.} for the two rules we are interested on -PROP and CEL- in terms of attribution rules. We give an interpretation of the rule as an attribution mechanism when it makes sense and discuss its validity as a property in the subclass defined by attribution problems otherwise.

The following properties, which are essential for characterizing CEL and PROP, are no longer valid in the domain ${\cal F}$, since some of the involved bankruptcy problems involved do not belong to the class ${\cal F}$ of bankruptcy problems compatible with attribution situations. In particular, {\em path Independence} property is essential for characterize both rules.

\begin{itemize}
    \item {\bf Composition. } For all $N$, all $(E, c)\in {\cal B}^N$ and all $E_1, E_2\in \mathbb{R}_{++}$ such that $E_1+E_2=E$, $R(E, c)=R(E_1, c)+R(E_2, c-R(E_1, c))$.


\item {\bf Path Independence. } For all $(E, c)\in {\cal B}^N$ and for all $E'>E$, $R(E, c)=R(E, R(E', c))$.

\item {\bf Consistency.} For all $N$, all $(E, c)\in {\cal B}^N$, all $S\subset N$ and all $i\in S$ we have: $R_i(N, E, c)= R_i(S, \sum_{i\in S}R_i(N, E, c), c_S)$, where $c_S=(c_i)_{i\in S}$.


\item {\bf SDU Self-duality. } For all $N$, all $(E, c)\in {\cal B}^N$, $R(E, c)=c-R(D, c)$.
SDU introduces a principle of symmetry in the behaviour of the solution with respect to awards and losses, and it is a crucial property of PROP rule. It says that the same principle is to be applied if we think of $(E, c)$ either a distribution problem or as a rationing scheme.
\end{itemize}

On the contrary, all properties of the following list do make sense as properties for attribution rules. We give their interpretation in these terms.

\begin{itemize}

\item {\bf IND Individual rationality. } For all $N$, and all $(E, c)\in {\cal B}^N$, $0\leq R_i(E, c) \leq c_i$, for all $i\in N$, establishing that we do not attribute a channel less than zero and no more than the sum of the total values in which it has had some participation. That is, no channel will be attributed by any value generated by those users that were not exposed to the Ad on it.
\item {\bf EFF Efficiency. } For all $N$, and all $(E, c)\in {\cal B}^N$, $sum_{i\in N} R_i(E,c)=E$. That is, all the value produced is attributed to the channels.

\item {\bf ETE Equal treatment of equals. } For all $N$, $(E, c)\in {\cal E}^N$ and for all $i,j\in N$, $c_i=c_j$ implies $R_i(E, c)=R_j(E, c)$.

\noindent For all $(B, c)\in {\cal B}^N$ and for all $i,j\in N$, $f(S\cup i)=f(S\cup j), \forall S\subseteq N\backslash\{i, j\}$ implies $R_i(B, c)=R_j(B, c)$.

As a property of an attribution rule is stronger than symmetry. Note that two channels can have the same claims but they could be distinguishable because they can differ in the values of the combinations they belong.











\item {\bf EXC Exclusion. } For all $N$, and all $(E, c)\in {\cal B}^N$, if $c_i\leq D/n$ then $R_i(B, c)=0$. Thus, the attribution will concentrate on that channels with the highest values.



One cannot claim more than there is; the excess is irrelevant.



\item {\bf CMR Composition from minimal rights. } For all $N$, all $(E, c)\in {\cal B}^N$, $R(E, c)=m(E, c)+R(E-\sum_{i\in N} m_i(E, c), c-m(E, c))$, where $m_i(E, c)=\max \{0, E-\sum_{j\neq i}c_j\}$ is the minimal right of player $i$.

The minimal right of a channel represents the amount of the global benefit that is left to him when the claims of all other channels are fully satisfied provided this amount is non negative. And it is taken to be zero otherwise. CMR assures to each agent their minimal right before assigning the remaining benefit considering the new claims.

In this case, it is no so obvious that the derived bankruptcy problem $(E-\sum_{i\in N} m_i(E, c), c-m(E, c))\in {\cal F}^N$. Let us check it:

First, it is obvious that $\sum_{i\in N} (c_i-m_i(E, c))\geq E-\sum_{i\in N} m_i(E, c)$.

Second, we must prove that $\max_{i\in N} c_i-m_i(E, c)\leq E-\sum_{i\in N} m_i(E, c)$.
Assume, without loss of generality that claims are ordered in decreasing order, $c_1\geq c_2\geq \cdots \geq c_n$, therefore $c_1=\max_i c_i$. As
\[
\sum_{i\neq 1} c_i\leq \sum_{i\neq j} c_i, \quad \forall j\in N,
\]
this implies that $m_i$ is a non increasing function in $i$. Then, if there exists $m_i\neq 0$, $m_1\neq 0$.

If $m_j>0$, $c_j-m_j=c_j-E+\sum_{i\neq j} c_i= D-E$. If for a given $j$, $m_j=0$, we must check that $c_1-m_1\geq c_j$. Assume $m_1>0$ (otherwise is trivial), $c_1-m_1=D-E=\sum_{i\in N} c_i -E\geq c_j$ because $m_j=0$ implies that $E-\sum_{i\neq j}c_i\leq 0$.

Therefore, we must check only that $E-\sum_{i\in N} m_i\geq c_1- m_1$. Let $k$ be such that $m_i>0$ for all $i=1,..,k$, and $m_i=0$ for $i> k$. The cases $k=0, 1$ are trivial. For $k\geq 2$ some simple calculations allows us to obtain that:
$$
E-\sum_{i\in N} m_i= (k-1)D+ \sum_{j=k+1}^n c_j
$$
Then, $E-\sum_{i\in N} m_i -(c_1- m_1)= (k-1)D+ \sum_{j=k+1}^n c_j - D= (k-2)D+ \sum_{j=k+1}^n\geq 0$.

\end{itemize}

The two attribution rules we are considering, PROP and CEL, verify IND, EFF and ETE, whereas  EXC and CMR are only satisfied by CEL.

\subsection{The repetition case}

Now we consider the case in which we not only record the visited channels but the number of times a channel appears in a path. For dealing with this case, we consider again the introduction of artificial players as in section \ref{No_OrderRep}.

In this context, we shall verify for CEL and PROP solutions that {\it ceteris paribus} the repetition of a channel favors it. First, a result about their behaviour when irrelevant claimants are disregarded is needed.

Let us define {\bf IPL Irrelevant players property}. We say that a rule $R$ verifies IPL if, for any bankruptcy problem $(N, E, c)$, and any player $k\in N$ with $R_k(N,E,c)=0$, the solution $R'$ of the reduced problem $(N\setminus \{k\}, E, c_{-k})$ is such that $R'_i=R_i$,  for all $i\in N\setminus \{k\}$.

\begin{proposition} The CEL and PROP attribution rules verify irrelevant players property.
\label{CEL-PROP-IPL}
\end{proposition}

\begin{proof}
Trivially, $PROP_i=0$ if, and only if, $c_i=0$, and thus PROP verifies IPL.

To prove that $CEL$ verifies IPL, we first we prove that the property is correctly defined for CEL, i.e. that the reduced problem is still a bankruptcy problem: $C-c_k=C_{-k}\geq E$.

If $CEL_k(E,c)=0$, then $c_k\leq (C-E)/n=(C_{-k}+c_k-E)/n$ holds. Thus $C_{-k}-E\geq (n-1)c_k\geq 0$, for all $n\geq 1$.

Now, let $\lambda$ the solution of the original problem defining $CEL(N,E,c)$, i.e.
 $CEL_i=\max \{0, c_i-\lambda\}$, for all $i\in N$ and
 $$
 E=\sum_{i\in N} \max \{0, c_i-\lambda\}=\sum_{\substack{i\in N\\ CEL_i(N,E,c)>0}} \max \{0, c_i-\lambda\} .
 $$
 Therefore, $\lambda$ is also the solution for the reduced problem defining $CEL(N\setminus \{k\}, E, c_{-k})$. Moreover, the whole set of irrelevant players can be removed without changing the proposed share.
\end{proof}

\begin{proposition} The CEL and PROP attribution rules verify monotonicity with respect to channel repetition.
\end{proposition}

\begin{proof}
Without loss of generality, we will prove monotonicity in the following case.

Let $( N,P(N), f )$ be an attribution problem with associated bankruptcy problem given by $(N, F, c)$. Assume that  a new path $p^r$ substitutes a path $p$ by repeating once a player $i\in p$ without changing its value $f(p^r)=f(p)$. Then, two new players are added to the problem in such a way that player $i$ is substituted by player $i_1$ in all paths and its repetition is substituted in path $p^r$ by $i_2$. The corresponding new bankruptcy problem is $(N^r, F, c^r)$, where $N^r=N\setminus i \cup\{i_1, i_2\}$, $c^r_j=c_j$, for all $j\in N^r$, $j\neq i_2$ and $c^r_{i_2}=f(p^r)$.

Let us start with the CEL rule. Then, we will prove that
\begin{equation}
        CEL_i(B)\leq CEL_{i_1}(B^r)+CEL_{i_2}(B^r),
        \label{cond:monotonicityCEL}
\end{equation}
where $B$ denotes the original bankruptcy problem $(N,F,c)$ and $B^r$ the new one $(N^r,F^r,c^r)$. We shall distinguish two cases:
\begin{itemize}
\item
Case 1: If $CEL_{i_2}(B^r)=0$ then, by IPL property, $CEL_{i_1}(B^r)=CEL_i(B)$ holds.
\item
Case 2: Otherwise, $CEL_{i_2}(B^r)>0$, then $c^r_{i_2}=f(p)> D^r/(n+1)$, where the deficit $D^r$ for the new bankruptcy problem is $D^r=D+f(p)$. Thus, $f(p)>D/n$ and  $\frac{D+f(p)}{n+1}\geq \frac{D}{n}$. Therefore, $c_j\geq \frac{D+f(p)}{n+1}\geq \frac{D}{n}$, and this implies that every excluded channel in $B$ is also an excluded channel in $B^r$. We show that this also implies that every irrelevant channel in $B$ is also irrelevant in $B^r$. Let $E(B)\subseteq N$ be the set of excluded players in $B$. Then, $c^r_{i_2}=f(p)> (D^r-\sum{j\in E(B)} c_j)/(\vert N\setminus E(B)\vert +1)$ since $CEL_{i_2}(B^r)>0$. Thus,
$$
c_k \geq \frac{D^r-\sum_{ j\in E(B) } c_j}{(\vert N\setminus E(B)\vert +1) } >
\frac{D-\sum_{ j\in E(B) } c_j}{\vert N\setminus E(B)\vert }.
$$
Thus every irrelevant non excluded player in $B$ is also an irrelevant player in $B^r$.

Let $R(B)$ be the set of relevant channels in $B$, $i\in B$ and $CEL_j(B)=c_j-\lambda$, for all $j\in R(B)$ with $E=\sum_{j\in R(B)} (c_j-\lambda)$. We consider again two different cases:
\begin{itemize}
    \item If there exists a channel $k\in R(B)$ which turns out to be irrelevant in $B^r$, then $c_k-\lambda>0>c_k-\lambda^r$, where $\lambda$ and $\lambda^r$ solve, respectively, the problems defining $CEL(B)$ and $CEL(B^r)$. Thus, $\lambda^r>\lambda$ and therefore $CEL_j(B)\geq \max\{ 0, c_j-\lambda^r\}=CEL_j(B^r)$ for all $j\in R(B)$ and condition \eqref{cond:monotonicityCEL} holds taking into account that
$$
\sum_{\substack{j\in R(B)\\ j\neq i}} CEL_j(B) + CEL_i(B)=E= \sum_{\substack{j\in R(B)\\ j\neq i}} CEL_j(B^r) +CE_{i_1}(B^r)+CEL_{i_2}(B^r).
$$
\item
Otherwise, the set $R(B^r)=R(B)\setminus \{i\} \cup \{ i_1,i_2\}$. We will prove that $\lambda^r\geq \lambda$. Let us suppose that $\lambda^r < \lambda$, then it holds:
$$
E=\sum_{\substack{j\in R(B)\\ j\neq i}} CEL_j(B) + CEL_i(B) <  \sum_{\substack{j\in R(B)\\ j\neq i}} CEL_j(B^r) +CE_{i_1}(B^r) \leq E,
$$
which is a contradiction. Therefore, $\lambda^r\geq \lambda$ and the reasoning of the previous case applies.
\end{itemize}
\end{itemize}

In order to prove $PROP_i(B)\leq PROP_{i_1}(B^)r+PROP_{i_2}(B^r)$, it is only necessary to check that:
\begin{equation*}
    \frac{c_iE}{C}\leq \frac{c_iE}{C+f(p)}+\frac{f(p)E}{C+f(p)},
\end{equation*}
which clearly holds since $c_i\leq C$
%
\end{proof}

\subsection{Order relevant case}

Under the same framework that in section \ref{Order} (order players, no repetition), the results obtained for the Shapley value of the sum game can only be replicated for the PROP solution but not in general for the CEL solution, i.e. PROP can be order player decomposed but no CEL.

\begin{example}
\label{ExBankOrder}

Given the campaign data in table \ref{DataBank1}, the CEL and PROP solutions appears in table \ref{SolBank1}. The corresponding order data and solutions are in tables \ref{DataBankOrder1} and \ref{SolBankOrder1}, respectively. Thus, we can check that PROP verifies decomposition with respect order players but not CEL:


$CEL_1=160/3 > CEL_{1_1}+CEL_{1_2}= 50$

$CEL_2=10/3 < CEL_{2_1}+CEL_{2_3}= 10$

$CEL_3=130/3 > CEL_{3_1}+CEL_{3_2}= 40$

In this example, the main reason for this result is that in the decomposition of player $2$, the order player $2_1$ becomes irrelevant with respect CEL rule.

\begin{table}[h]
\caption{Campaign data}
\label{DataBank1}       
\begin{tabular}{lc}
\hline\noalign{\smallskip}
Path $p$ & KPI value $f(p)$
\\
\noalign{\smallskip}\hline\noalign{\smallskip}
$(1)$ & 20  \\
$(1,3)$ & 40  \\
$(3,1,2)$ & 30  \\
$(2,3)$ & 10  \\
\noalign{\smallskip}\hline
\end{tabular}
\end{table}

\begin{table}[h]
\caption{Bankruptcy solutions without order }
\label{SolBank1}       
\begin{tabular}{cccc}
\hline\noalign{\smallskip}
Players  & $1$ & $2$  & $3$
\\
\noalign{\smallskip}\hline\noalign{\smallskip}
CEL & 160/3 & 10/3  &  130/3 \\
PROP & 900/21 & 400/21 & 800/21 \\
\noalign{\smallskip}\hline
\end{tabular}
\end{table}

\begin{table}[h]
\caption{Campaign data with order players }
\label{DataBankOrder1}       
\begin{tabular}{lc}
\hline\noalign{\smallskip}
Path $p$ & KPI value $f(p)$
\\
\noalign{\smallskip}\hline\noalign{\smallskip}
$(1_1)$ & 20  \\
$(1_1,3_2)$ & 40  \\
$(3_1,1_2,2_3)$ & 30  \\
$(2_1,3_2)$ & 10  \\
\noalign{\smallskip}\hline
\end{tabular}
\end{table}

\begin{table}[h]
\caption{Bankruptcy solutions with order }
\label{SolBankOrder1}       
\begin{tabular}{ccccccc}
\hline\noalign{\smallskip}
Players  & $1_1$ & $1_2$  & $2_1$ & $2_3$ & $3_1$ & $3_2$
\\
\noalign{\smallskip}\hline\noalign{\smallskip}
CEL & 40 & 10  &  0 & 10 & 10 & 30 \\
PROP & 600/21 & 300/21 & 100/21 & 300/21 & 300/21 & 500/21  \\
\noalign{\smallskip}\hline
\end{tabular}
\end{table}

\end{example}

When introducing order players (without repetition) in bankruptcy problems, two things are preserved:

\begin{itemize}
    \item The claims of the order players $i_1,...,i_{p_i}$ corresponding to a player $i\in N$ verify: $c_1+\cdots+c_{p_i}=c_i$. Therefore, the deficit $D$ remains just the same. This is the main difference with respect the repetition case.
    \item The extended bankruptcy problem with order players continues to be in the ${\cal F}^N$ class.
\end{itemize}

 Then, the analysis of the changes of a bankruptcy rule for the attribution problem with order players is just an splitting-proof analysis for a general bankruptcy problem. See, for instance, Ju (2003), Ju et al. (2007) and Moreno-Ternero (2007).

 \begin{definition}
 A rule $R$ is splitting-proof if for all $(N, E, c), (N', E, c') \in {\cal B}$, with $N\subset N'$, and such that there is some $i\in N$ such that $c_i = c'_i + \sum_{j\in N'\setminus N} c'_j$ and for each $j \in N\setminus \{i\}$, $c'_j= c_j$ then $R_i(N,E, c) \geq R_i(N',E, c') + \sum_{j\in N'\setminus N} R_j(N',E, c')$.
 \end{definition}

 In particular, PROP is a splitting-proof rule, but moreover it is straightforward to check the following proposition:

\begin{proposition}
If no repetition occurs, PROP rule verifies decomposition with respect to order players.
\end{proposition}

CEL is also a splitting-proof rule (Moreno-Ternero, 2007) under the above definition. It is more difficult to obtain general results when all players can split as in the order case of a bankruptcy attribution problem. However a rather restrictive result can be given:

\begin{proposition}
\label{CELOrdDec}
Given an attribution bankruptcy problem $(N, F, c)\in {\cal F}^N$, if we assume that all players in $N$ and all their correspondent order players are relevant, then $CEL_i\geq \sum_{j=1}^{p_i} CEL_{i_j}$  if, and only if, $p_i/ \sum_{\ell\in N} p_{\ell}\geq 1/n$, where $p_i$ is the number of relevant order players corresponding to player $i\in N$.
\end{proposition}

\begin{proof}
We have that $CEL_i=c_i-D/n$ and $CEL_{i_j}=c_{i_j}-D/\sum_{\ell\in N} p_{\ell}$ for all $j=1,...,p_i$, taking into account that $c_i=\sum_{j=1}^{p_i} c_{i_j}$, thus $CEL_i\leq \sum_{j=1}^{p_i} CEL_{i_j}$ if and only if $p_i/ \sum_{\ell\in N} p_{\ell}\leq 1/n$.
\end{proof}

\noindent That is, channel $i$ benefits from order player decomposition if its proportion of order players is under the mean $1/n$.

As we have seen in the above example, CEL does not verify in general the decomposition property with respect order players. However, in some particular case, it can occur as an immediate corollary of the above proposition:

\begin{corollary}
Under the same conditions of proposition \ref{CELOrdDec}, if $p_i=p_{\ell}$ for all $i,{\ell}\in N$, CEL verifies decomposition with respect order players.
\end{corollary}

\begin{example}
\label{ExBankOrder2}

Given the campaign data in table \ref{DataBank2}, the CEL and PROP solutions appears in table \ref{SolBank2}. The corresponding order data and solutions are in tables \ref{DataBankOrder2} and \ref{SolBankOrder2}, respectively. In this example there is not irrelevant players and the number of order players are the same. We can check that CEL does verify the decomposition:

$CEL_1=88/3 = CEL_{1_1}+CEL_{1_2}$.

$CEL_2=76/3 = CEL_{2_1}+CEL_{2_3}$.

$CEL_3=136/3 = CEL_{3_1}+CEL_{3_2}$.

\begin{table}[h]
\caption{Campaign data}
\label{DataBank2}       
\begin{tabular}{lc}
\hline\noalign{\smallskip}
Path $p$ & KPI value $f(p)$
\\
\noalign{\smallskip}\hline\noalign{\smallskip}
$(1)$ & 14  \\
$(1,3)$ & 20  \\
$(3,1,2)$ & 36  \\
$(2,3)$ & 30  \\
\noalign{\smallskip}\hline
\end{tabular}
\end{table}

\begin{table}[h]
\caption{Bankruptcy solutions without order}
\label{SolBank2}       
\begin{tabular}{cccc}
\hline\noalign{\smallskip}
Players  & $1$ & $2$  & $3$
\\
\noalign{\smallskip}\hline\noalign{\smallskip}
CEL & 88/3 & 76/3 &  136/3 \\
PROP & 7000/222 & 6600/222 & 8600/222 \\
\noalign{\smallskip}\hline
\end{tabular}
\end{table}

\begin{table}[h]
\caption{Campaign data with orders }
\label{DataBankOrder2}       
\begin{tabular}{lc}
\hline\noalign{\smallskip}
Path $p$ & KPI value $f(p)$
\\
\noalign{\smallskip}\hline\noalign{\smallskip}
$(1_1)$ & 14  \\
$(1_1,3_2)$ & 20  \\
$(3_1,1_2,2_3)$ & 36  \\
$(2_1,3_2)$ & 30  \\
\noalign{\smallskip}\hline
\end{tabular}
\end{table}

\begin{table}[h]
\caption{Bankruptcy solutions with order }
\label{SolBankOrder2}       
\begin{tabular}{ccccccc}
\hline\noalign{\smallskip}
Players   & $1_1$ & $1_2$  & $2_1$ & $2_3$ & $3_1$ & $3_2$ \\
\noalign{\smallskip}\hline\noalign{\smallskip}
CEL & 41/3 & 47/3   &  29/3 & 47/3 & 47/3 & 89/3 \\
PROP & 3400/222 & 3600/222 & 3000/222 & 3600/222 & 3600/222 & 5000/222  \\
\noalign{\smallskip}\hline
\end{tabular}
\end{table}

\end{example}

\section{Conclusions}
\label{Con}

In this paper, we have addressed the attribution problem that arise when the total benefits obtained by a marketing campaign must be distributed among the different advertising channels involved in the campaign, which is nowadays a cornerstone of any multi-channel marketing strategy. We have proposed to rely on Game Theory to propose new methods, as those based on bankruptcy problems, for the assessment of the benefits and/or to analyze rigorously the properties of the derived attribution rules.

We have essentially analyzed two kinds of mechanisms, one of them based on the Shapley value of an appropriate TU game, and the other one based on bankruptcy problems. The first one was already considered in Morales (2016), Cano-Berlanga et al. (2017) and Zhao et al. (2018) for the simplest case in which nor order neither repetition play a role. We have extended its use to these more general cases and we have developed a thorough analysis of its properties in terms of attribution mechanisms. From a practical point of view the Shapley attribution has many advantages: 1) The characteristic function of the sum game is conceptually well defined and has good properties as the monotonic and superadditive properties. Its definition is simple and easy to understand and explain; 2) The calculation method is very simple and also helps to understand the method: the value attributed to a string is the sum of the aliquot part of the value of each combination to which it belongs; 3) Its additivity allows to manage jointly a batch of related campaigns; and 4) Monotonicity properties serve as incentives for the channels to increase their presence.

The approach based on Bankruptcy problems is as far as we know new. We have proposed an appropriate bankruptcy problem to deal with attribution problems which has an intuitive and easy to explain definition. Among the existing bankruptcy rules, it is the CEL rule that establishes an alternative view more different to the Shapley value of the sum game. In this case, the exclusion property is fundamental: it gives zero value to very weak channels, if any, and tends to concentrate the attribution in the channels that belong to the highest valued combinations. As it happens with the Shapley attribution rule, the CEL rule is also simple to calculate. However, additivity is lost and therefore the shares of the aggregated benefit obtained through a batch of  campaigns must be treated as a new sharing problem and can not be decomposed in terms of the shares of each campaign. Taking into account that the value $f(i)$ is obtained only by channel $i$, it would be interesting to consider extended bankruptcy problems in which each agent has an {\em objective entitlement} besides her claim. These models are analyzed in Pulido et al. (2002). From a game theoretical point of view, the general own interest of the subclass ${\cal F}$ must be remarked. To obtain an axiomatic characterization of CEL and PROP rule for this subclass must be tackled in subsequent works.

\end{document}